\documentclass[12pt,a4paper]{article}

\usepackage{amsmath}
\usepackage[utf8]{inputenc}
\usepackage{amssymb}
\usepackage{bbm}
\usepackage{comment}
\usepackage{geometry}
\usepackage{natbib}
\usepackage{hyperref}
\usepackage{color}
\usepackage{setspace} 

\usepackage{algorithm}
\usepackage{algpseudocode}

\newtheorem{theorem}{Theorem}

\newtheorem{lemma}{Lemma}
\newtheorem{proposition}{Proposition}

\newtheorem{conjecture}{Conjecture}
\newtheorem{example}{Example}

\newenvironment{proof}[1][Proof]{\noindent\textbf{#1. }}{\ \rule{0.5em}{0.5em}}

\begin{document}

\begin{titlepage}
\title{
	Vote Delegation with Unknown Preferences
}

\author{
	Hans Gersbach\\
	\normalsize CER-ETH and CEPR\\
	\normalsize Z\"{u}richbergstrasse 18\\
	\normalsize 8092 Zurich, Switzerland\\ 
	\normalsize \href{mailto:hgersbach@ethz.ch}{hgersbach@ethz.ch}
	\and
	Akaki Mamageishvili\\
	\normalsize CER-ETH\\
	\normalsize Z\"{u}richbergstrasse 18\\
	\normalsize 8092 Zurich, Switzerland\\ 
	\normalsize \href{mailto:amamageishvili@ethz.ch}{amamageishvili@ethz.ch}
	\and
	Manvir Schneider\\
	\normalsize CER-ETH\\
	\normalsize Z\"{u}richbergstrasse 18\\
	\normalsize 8092 Zurich, Switzerland\\ 
	\normalsize \href{mailto:manvirschneider@ethz.ch}{manvirschneider@ethz.ch}
	}

\date{Last updated: \today}

\maketitle

\begin{abstract}
We examine vote delegation when preferences of agents are private information. One group of agents (delegators) does not want to participate in voting and abstains under conventional voting or can delegate its votes to the other group (voters) who decide between two alternatives. We show that free delegation favors minorities, that is, alternatives that have a lower chance of winning ex-ante. The same occurs if the number of voting rights that actual voters can have is capped. When the number of delegators increases, the probability that the ex-ante minority wins under free and capped delegation converges to the one under conventional voting---albeit non-monotonically. 
Our results are obtained in a private value setting but can be readily translated into an information aggregation setting when voters receive a signal about the ``correct" alternative with some probability.
\end{abstract}
\thispagestyle{empty}
\end{titlepage}

\pagebreak \newpage

\setcounter{page}{2}
\onehalfspacing

\section{Introduction}

\subsection{Motivation and Results}

While a particular type of vote delegation in the form of proxy voting in shareholder meetings has a long tradition, the issue whether and how vote delegation should be allowed in other environments has become a prominent theme in scientific and public discussions. ``Liquid democracy", for instance, is a concept entailing the right for citizens to delegate voting rights to other citizens. This has been advocated by several political parties.\footnote{It is used internally by the German Pirate party or the Demoex party in Sweden. It is implemented in several online platforms such as LiquidFeedback \url{www.liquidfeedback.org}  (retrieved on February 18, 2021), and Google Votes \url{https://www.tdcommons.org/cgi/viewcontent.cgi?article=1092&context=dpubs_series} (retrieved on February 18, 2021).  
A further reference is \cite{escoffier-gilbert}.} It is an important way to govern blockchains, see~\cite{tezos} and~\cite{concordium}.
A priori, it is not obvious how liquid democracy compares to standard voting procedures such as conventional voting. It is thus useful to compare liquid democracy to standard voting in specific environments.

In this paper, we examine how voting outcomes are affected when vote delegation is allowed, but delegators do not know the preferences of the proxy to whom they delegate. In particular, we study the following problem. A polity---say a jurisdiction or a blockchain community---decides between two alternatives $A$ and $B$. A fraction of individuals does not want to participate in the voting. Henceforth, these individuals are called ``delegators". They abstain in traditional voting and delegate their vote in liquid democracy. 

There are many reasons why individuals do not want to participate in voting. Either they want to avoid facing the costs of informing themselves or do not want to incur the cost of voting. 

In the context of blockchains and proof-of-stake, individuals may have the option to delegate their stake to other individuals, thereby participating in the rewards from validating transactions without having to run the validation software themselves. Several blockchains have either implemented or announced such schemes (see Tezos\footnote{See \cite{tezos}}, Palkadot\footnote{\url{https://polkadot.network/} (retrieved on October 5, 2021)}, Cardano\footnote{\url{https://docs.cardano.org/core-concepts/delegation} (retrieved on October 5, 2021)} or Concordium\footnote{See \cite{concordium}}).

Our central assumption is that delegators do not know the preferences of the individuals who will vote, that is, whether these individuals favor $A$ or $B$. This is obvious in the context of blockchains in which delegators do not know whether voters are honest (interested in the validation of transactions) or dishonest (interested in creating invalid transactions or in disrupting the system as a whole). For standard democracies, our assumption represents a polar case.

We call individuals willing to vote ``voters", individuals favoring $A$ ($B$) ``$A$-voters" (``$B$-voters"), and we call the group of the  $A$-voters ($B$-voters) ``$A$-party" (``$B$-party"). From the perspective of delegators, of voters themselves and of any outsider who might want to design vote delegation, the probability that a voter favors $A$ is some number $p$ ($0<p<1$). If $p<\frac{1}{2}$, the chances of alternative $A$ to win under standard voting is positive but below $\frac{1}{2}$. We call the $A$-party the ``ex-ante minority" in such cases.

We compare three voting schemes. First, under standard voting, delegation is not allowed and thus delegators abstain. Second, with free vote delegation, arbitrary delegation is allowed. Since delegators do not know the preferences of voters and all voters are thus alike for them, every voter has the same chance to obtain a vote from delegators. 
We thus model free delegation as a random process in which delegated voting rights are randomly and uniformly assigned to voters. This can be achieved in one step where delegators know the pool of voters, or by transitive delegation, in which each individual can delegate his/her voting right to some other individual, who, in turn, can delegate the accumulated voting rights to a next individual. This process stops when an individual decides to use the received votes and exert all voting rights received.
Third, we introduce capped vote delegation. With capped vote delegation, the number of votes an individual can receive and can cast is limited. Essentially, if a voter has reached the cap of voting rights through delegation, s/he can't receive delegated votes anymore. Thus, only voters who have not reached the cap yet will be able to receive further voting rights.

In our analysis, we model the total number of voters by a random variable with distribution $F$, whereas $F$ can be any discrete distribution. 
We establish three main results: First, compared to standard voting, free delegation favors the minority.
This result, which requires an elaborated proof, shows that vote delegation cannot only reverse voting outcomes but increases the probability that the minority wins.

An essential assumption for this result is that the types of voters (and thus their preferences) are stochastically independent. By a counterexample, we show that delegation can favor the majority if this assumption is not made. Hence, a model which allows agents to delegate is {not} the same as adding noise to the model which always favors the minority. However, in the class of distributions where agents receive stochastically independent signals of their type, such as for instance for a Poisson distribution, delegation favors the minority.

Second, the same occurs with a capped delegation. That is, the capped delegation also increases the minority's probability to win. Numerical examples show that the increase is smaller than under free delegation. 
Third, when the fraction of delegators increases compared to voters, the probability that minorities win under free or capped delegation converges to the same probability as under standard voting. However, we see in the case for a Poisson distribution that the convergence is not monotonic. 
The results show that outcomes \textit{with delegation} may significantly differ from \textit{standard} election outcomes and probabilistically violate the majority principle, i.e. a majority of citizens who is willing to cast a vote for a particular alternative may lose. This result might make blockchain companies hesitant to implement vote delegation, as it might raise the risk that dishonest agents become a majority through vote delegation.

\subsection{Application to Information Aggregation}

While we derive our results in the private value setting, the results can be directly applied to information aggregation in a common value setting. The standard modeling is that one alternative is preferred by everybody (henceforth the correct alternative), but individuals lack information which one is  correct. Each individual receives a signal which alternative is correct. The probability that the signal identifies the correct alternative is $p$ and it is assumed that $p > \frac{1}{2}$. Each individual votes for the alternative s/he thinks is correct. 

The Condorcet Jury Theorem states that adding more voters increases the probability that the correct alternative is chosen by the majority. If the number of voters converge to infinity, the probability of implementing the right alternative tends to $1$. 

With our approach, we can study how vote delegation affects the results. Suppose again that one group of individuals does not want to participate in voting. The reason may be the costs or the fact that they may have received a completely uninformative signal, i.e. $p= \frac{1}{2}$. 
Individuals of the other group have received stochastically independent signals which identify the correct alternative with probability $p$, $p > \frac{1}{2}$, and they vote according to their signals. 

Individuals in the first group either abstain under conventional voting or delegate their vote to some individual in the second group.\footnote{If delegators have received uninformative signals, $p= \frac{1}{2}$ and voted, this would be the worst outcome.} Our results imply that delegation lowers the chance that the correct alternative wins. Despite the fact that individuals with a signal that identifies the correct alternative with a probability greater than $\frac{1}{2}$ receive more votes to be cast, the probability that the correct alternative wins declines. The reason is that different votes do not reflect stochastically independent signals. Hence, abstention of all delegators maximizes the probability that the correct alternative achieves a majority of votes.

To illustrate the result, consider the following simple example. There are three voters, two delegators and the probability that each voter receives the correct signal is equal to $p$. 

Suppose first that delegators abstain. Then, the probability that majority voting implements the correct outcome is equal to $r:=p^3+{3\choose 2}p^2(1-p)$. Note that since $p>\frac{1}{2}$, we have $r>p$. That is, with three voters the probability that the correct alternative is implemented is higher than with only one voter. This is a manifestation of the wisdom of crowds and part of the steps to prove the Condorcet Jury Theorem.  

Suppose next that the two delegators delegate their votes uniformly at random to the three voters. Two cases are possible: In the first case, two voters obtain one additional vote while the third voter has only his own vote. In this case, the probability that the correct alternative is implemented by majority of voting remains equal to $r$. In the second case, one voter obtains both votes from the delegator. This happens with probability $1/9$. In this case, the probability that the correct alternative is implemented is equal to $p$, since one voter has three votes and his votes form a majority that is independent of how the other two voters vote. Taking both cases together, we observe that the expected probability that the correct alternative is implemented under delegation is $8/9r + 1/9p$ and thus strictly lower than $r$. Hence, delegation leads to worse outcomes.

\section{Related Literature}
Some recent papers shed new light on how vote delegation impacts outcomes. Important studies examining vote delegation from an algorithmic and AI perspective have developed several delegation rules that allow to examine how many votes a representative can and should obtain\footnote{See \cite{liquid_fluid} and \cite{kotsialou-riley}.} and whether the delegation of votes to neighbours in the network may deliver more desirable outcomes than conventional voting\footnote{See \cite{liquid_algorithmic}.}. The result of \cite{liquid_algorithmic} is similar to our result in spirit. The authors show that even with a delegation from a less informed to a better informed voter, the probability of implementing the right alternative in a generic network is decreasing. The setting studied in the paper is different, though, as it focuses on the information acquisition aspect of voting, with voters having the same preferences. 

Delegation games on networks were studied by~\cite{rational_delegation} and~\cite{escoffier-gilbert}.
~\cite{rational_delegation} identify conditions for the existence of Nash equilibria of such games, while~\cite{escoffier-gilbert} shows that in more general setups, no Nash equilibrium may exist, and it may even be $NP$-complete to decide \textit{whether} one exists at all. We adopt a collective decision framework and study how free delegation and capped delegation impact outcomes when delegators do not know the preferences of the individuals to whom they delegate.

There is recent work on representative democracy, e.g. \cite{PIVATO202052} and \cite{ijcai2019-1}. The former paper studies a model where legislators are chosen by voters and votes are counted according to the weight, which is the number of votes a representative possesses.  The show that the voting outcome is the same in a large election as for conventional voting. \cite{SohVoutsa} extended this model by using other forms of voting, such as \textit{Weighted Approval Voting}.
Furthermore there is work on \textit{Flexible Representative Democracy}, see e.g. \cite{ijcai2019-1}. In this model, first experts are elected and then voters allocate their vote among the representatives or vote directly.
In comparison to our model, we note that delegators delegate to any voter in the society. In Flexible Representative Democracy, random uniform delegation can be achieved by uniformly allocating the vote among the representatives.

\section{Model}

We consider a large polity (a society or a blockchain community) that faces a binary choice between two alternatives $A$ and $B$. There is a group of $m \in \mathbb{N}_+$ individuals who do not want to vote and either abstain under conventional voting or delegate their voting rights under vote delegation. An agent in this group is called a \textit{delegator}. The remaining population votes and is thus called \textit{voters}. Voters have private information about their preference for $A$ or $B$. Hence, when a delegator delegates her voting right to a voter, it is equivalent to uniformly and randomly delegating a voting right to \textit{one} voter. A voter prefers alternative $A$ ($B$) with probability $p$ ($1-p$), where $0<p<1$. Without loss  of generality, we assume $p>\frac{1}{2}$. Voters favoring $A$ ($B$) are called ``$A$-voters" (``$B$-voters") and the respective group is called ``$A$-party" (``$B$-party").
{Note that in the information aggregation setting this corresponds to voters receiving signal $A$ with probability $p$ and signal $B$ with probability $1-p$.}

We assume that the number of voters is given by a random variable with distribution $F$, that is, the probability that there are $i$ voters is equal to $F(i)$. We compare three voting processes:
\begin{itemize}
    \item Conventional voting: Each voter casts one vote.
    \item Free delegation: $m$ delegators randomly delegate their voting rights. Each voter casts all votes that is corresponding to her voting rights.
    \item Capped delegation: $m$ delegators delegate their voting rights randomly to voters. If a voter has reached the cap, all further voting rights are distributed among the remaining voters up to the cap. If all voters have reached their cap, superfluous voting rights are eliminated.
\end{itemize}

We start with conventional voting  and denote by $P(p)$ the probability that $A$ wins. It is calculated by the following formula:

\begin{equation}\label{pr_conv}
    P(p) = \sum_{i=0}^{\infty}F(i)\sum_{k=0}^{i}{i \choose k}p^k(1-p)^{i-k}g(k, i-k). 
\end{equation}

where $g(k,l)$ is the probability that $A$-voters are in majority if there are $k$ $A$-voters and $l$ $B$-voters, it is calculated in the following way: 

\begin{align*}
  g(k,l):=\begin{cases}
               1, & \text{if } k>l, \\
               \frac{1}{2}, & \text{if } k=l,\\
               0, & \text{if } k<l.
            \end{cases}    
\end{align*}

In case of free vote delegation, the probability that $A$ wins, denoted by $P(p,m)$,  is equal to:

\begin{equation}\label{pr_free_del}
    {P}(p,m) = \sum_{i=0}^{\infty}F(i)\sum_{k=0}^{i}{i \choose k}p^k(1-p)^{i-k}G(k, i-k,m). 
\end{equation}

where $G(k,l,m)$ denotes the probability that $A$-voters win if there are $k$ $A$-voters, $l$ $B$-voters and $m$ voters delegate their votes. The value is calculated in the following way:

\begin{equation*}
G(k,l,m)=
    \begin{cases}
        \frac{1}{2} & \text{if } k = l = 0,\\
        \sum_{h=0}^{m}{m \choose h}\left(\frac{k}{k+l}\right)^h\left(\frac{l}{k+l}\right)^{m-h}g(k+h, l+m-h) & \text{else}.
    \end{cases}
\end{equation*}

We now consider the delegation procedure when delegated votes are distributed randomly among the other voters, with the restriction that no voter obtains more than $c$ votes. If a voter already has $c$ votes, s/he is not allowed to receive more. We invalidate the rest of delegated votes, if there are any.
With cap $c$, the probability that $A$-voters have a majority, denoted as $P_c(p,m)$, is equal to

\begin{equation}\label{pr_cap_del}
    P_c(p,m) = \sum_{i=0}^{\infty}F(i)\sum_{k=0}^{i}{i \choose k}p^k(1-p)^{i-k}G_c(k, i-k,m). 
\end{equation}

where $G_c(k,l,m)$ denotes the probability that $A$-voters win if there are $k$ $A$-voters, $l$ $B$-voters and $m$ voters who delegate their votes uniformly, with an individual voter being allowed to have $c$ votes at most. The value is calculated in the following way:

\begin{equation*}
G_c(k,l,m)=
    \begin{cases}
        \frac{1}{2} &  \text{if } k = l = 0,\\
        \sum_{h=0}^{m}{m \choose h}\left(\frac{k}{k+l}\right)^h\left(\frac{l}{k+l}\right)^{m-h}g(\min\{k+h,ck\}, \min\{l+m-h, cl\}) & \text{else}.
    \end{cases}    
\end{equation*}

Since $ck$ is the maximum amount of votes majority voters can have, we take $\min\{k+h,ck\}$ to calculate the total number of votes for majority voters. Similarly, $cl$ is the maximum amount of votes minority voters can have. Therefore, we take a value $\min\{l+m-h, cl\}$ to calculate the total number of votes minority voters can have. 

For notational convenience, we denote a binomial random variable with parameters $n$ and $p$ by $Bin(n,p)$, for any $n\in \mathbb{N}$ and $p\in[0,1]$.

\section{Results}

To gain an intuition for the formal results, we start with numerical examples. For this we model the number of voters as a  Poisson random variable, with average $n$. That is, $F(i) = \frac{n^i}{e^n i!}$. To indicate that we are considering a specific distribution with a parameter, we include the parameter in the function $P()$. Thus, if we consider a Poisson distribution with parameter $n$, then we write $P(n,p)$ instead of $P(p)$ and $P(n,p,m)$ instead of $P(p,m)$. Accordingly, for the capped case.

This assumption is a standard tool to make the analysis of voting outcomes analytically tractable. Namely, by {\it decomposition property}, the number of $A$-voters is Poisson random variable with average $n\cdot p$ and the number of $B$-voters is Poisson random variable with average $n \cdot (1-p)$.  Moreover, these two random variables are independent. 
Poisson games were introduced in~\cite{poisson}.

\begin{table}[h!]
    \centering
    \begin{tabular}{|c|c|c|c|c|c|}
        \hline
         $n$ & $p$ & $m$ & $P(n,p)$ & $P(n,p,m)$ & $P_2(n,p,m)$  \\ \hline
         $20$ & $0.6$ & $1$ & $0.81413$ & $0.808443$ & $0.808443$\\ \hline
         $20$ & $0.6$ & $2$ & $0.81413$ & $0.804256$ & $0.804256$\\ \hline 
         $20$ & $0.6$ & $5$ & $0.81413$ & $0.796578$ & $0.796616$\\ \hline
         $20$ &  $0.6$ &  $10$ &  $0.81413$ &  $0.791246$ & $0.792627$ \\ \hline
         $20$ &  $0.6$ &  $300$ &  $0.81413$ &  $0.808516$ & $0.81413$ \\ \hline
    \end{tabular}
    \caption{Probabilities of $A$ winning under conventional voting $P(n,p)$, free delegation $P(n,p,m)$ and capped delegation $P_c(n,p,m)$ for cap $c = 2$.}
    \label{table:example}
\end{table}

Table~\ref{table:example} reveals that the likelihood of $A$ winning is smaller with free delegation than under conventional voting. The same occurs with capped delegation but is slightly less pronounced than with free delegation. We also observe that when $m$ increases, both the winning probabilities for free delegation and for capped delegation first decline and then start to converge to the probability under conventional voting. The pattern in Table~\ref{table:example} repeats for different values of $n$ and $p$. In the following, we prove formal results.

\subsection{Free Delegation}

We start with free delegation and show the following:

\begin{theorem}\label{thm}
$P(p,m)<P(p)$ for any $m\geq 1$ and $p>0.5$.
\end{theorem}
\begin{proof}
See Appendix \ref{app:proofs}.
\end{proof}

In other words, Theorem~\ref{thm} says that free delegation probabilistically favors the ex-ante minority, since the probability that the minority wins under vote delegation is $1-P(p,m)$, while the probability that the minority wins in conventional voting is equal to $1-P(p)$. A trivial observation is that if $m=0$ then both $P(p,m)$ and $P(p)$ coincide since there is no vote to be delegated.

We note that our result rests on the assumption that types (preferences) of voters are stochastically independent. By a counterexample we illustrate that delegation may favor the majority in alternative uncertainty models. Suppose, for instance, there exists a joint distribution $J$ over the pair of natural numbers $(a,b)$, where $a$ denotes the number voters for alternative $A$ and $b$ denotes the number of voters for alternative $B$. Suppose the probability that there are more $A$-voters than $B$-voters is strictly larger than $\frac{1}{2}$. We next construct a distribution $J$ such that the probability that alternative $A$ is winning increases when there is delegation.
\begin{example}
We assume that the joint distribution $J$ is the product of two distribution functions $f$ and $g$, i.e.  $J(a,b) = f(a)\cdot g(b)$. Suppose that  $f(4)=0.6, f(1) = 0.4$ and $g(2) = 1$. Then $J(4,1) = 0.6$ and $J(1,2) = 0.4$ and $m=1$. The probability that alternative $A$ wins under conventional voting when the delegator abstains is $0.6$. The probability that alternative $A$ wins after the delegator gives his voting right to one of the voters is\footnote{One can construct examples that the increase of the likelihood that the majority wins is large and in the limit as  high as 0.25.}

$$0.6+f(1)\cdot \frac{1}{2} \cdot \frac{1}{3}>0.6.$$  
\end{example}

\subsection{Capped Delegation}\label{sc:capped_delegation}
The idea of capped delegation is to cap the number of votes that an individual voter can aggregate. 
Capped delegation is done according to Algorithm~\ref{alg:cap}.
\begin{algorithm}[H]
\caption{Capped Delegation}\label{alg:cap}
\begin{algorithmic}
\Require $m \geq 1$, $c \geq 1$ \Comment{Number of delegators $m$; Cap $c$}
\Require $i$ \Comment{Number of voters (There are $i$ voters with probabilty $F(i)$)}

\State $reachedCap \gets 0$
\State $v_1,...,v_i \gets 1$ \Comment{Each voter $j$ has $v_j$ votes}
\While{$m \neq 0 \textbf{ and } reachedCap \neq i$}
\State \text{Pick a random voter $j \in \{1,...,i\}$ (with probability $1/i$)}
\If{$v_j < c$}
    \State $v_j \gets v_j + 1$
    \State $m \gets m-1$
    \If{$v_j = c$}
        \State $reachedCap \gets reachedCap +1$
    \EndIf
\EndIf
\EndWhile
\end{algorithmic}
\end{algorithm}
The algorithm requires the number of delegators $m$, the cap $c$ and the number of voters which is a random variable with distribution $F$. In the beginning voters have 1 vote each (i.e. $v_i = 1$ for all $i$) and no voter has reached the cap and hence the variable $reachedCap$ is set to 0. As long as there are votes to be delegated, that is, $m\neq 0$ and there is still some voter who has not reached the cap, votes are delegated as follows: First a voter $j$ is chosen uniformly and at random. Then, if this voter has not reached the cap (i.e. $v_j<c$) then the voter gets an additional vote and the number of votes to be delegated is decreased by one. If a voter reaches the cap we increase $reachedCap$ by one.

We establish the following results in the case of capped delegation:

\begin{theorem}\label{thm2}
$P_c(p,m)<P(p)$ for any $m\geq 1$, $c>1$ and $p>0.5$.
\end{theorem}
\begin{proof}
See Appendix \ref{app:proofs}.
\end{proof}

Theorem~\ref{thm2} shows that capped delegation probabilistically favors the ex-ante minority. We note that if $c=1$ or $m=0$ then both $P_c(p,m)$ and $P(p)$ coincide since every voter can have at most one vote or there is no vote being delegated in the other case.

\begin{conjecture}\label{thm3}
$P(p,m) \leq P_c(p,m)$ for any $m\geq 1$, $c\geq 1$ and $p>0.5$.
\end{conjecture}

Conjecture~\ref{thm3} shows that capped delegation is better for the ex-ante majority than free delegation. Together with Theorem~\ref{thm2}, it implies that capped delegation is in-between standard voting (no delegation) and free delegation, with respect to the probability of majority winning. 
Note that if $c\geq m+1$, then both $P(p,m)$ and $P_c(p,m)$ coincide since the cap is higher than the total amount of delegated votes. For $c=1$, we have strict inequality following from Theorem~\ref{thm} because $P_1(p,m)=P(p)$. Finally, for $m=1$ and $c>1$, again both $P(p,m)$ and $P_c(p,m)$ coincide. Numerical calculations support the statement of the conjecture for the remaining cases of $c$ and $m$, see e.g., Table~\ref{table:example}.

\subsection{Asymptotic Behavior of Delegation}

In this section, we consider three cases of asymptotic behavior of delegation. In the first, the probability that a voter is in the majority is large. In the second, the average total number of voters is large. Both cases describe typical assumptions in real-world situations. In the first, one assumes there are very few voters in the minority, for example malicious voters in the context of blockchains. In the second, we look at a large electorate, with fixed shares of majority and minority voters. In the third, we show the convergence of  free delegation towards conventional voting as the number of delegators goes to infinity. 

First, we obtain a rather straightforward proposition:

\begin{proposition}\label{prop: lim_p_to_1}
For any fixed natural number $m \geq 1$, we have that $\lim_{p\rightarrow 1}P(p,m)=1$. 
\end{proposition}
\begin{proof}
See Appendix \ref{app:proofs}.
\end{proof}

Proposition~\ref{prop: lim_p_to_1} shows that for any number of delegators the probability that alternative $A$ wins is 1 when the probability that a voter is an $A$-voter goes to 1.

Next, we consider a large electorate, which we model using \textit{Poisson random variable} with parameter $n$. It is straightforward that with a constant number of delegators, the probability that the majority wins converges to one as the total population size converges to infinity. We show that the same holds even if there are arbitrarily many delegators.

\begin{proposition}\label{prop: lim_n_to_inf}
$\lim_{n\rightarrow \infty}P(n,p,m)=1$ for any fixed $p>0.5$ and any $m$, where $m$ can even depend on $n$.
\end{proposition}

\begin{proof} 
See Appendix \ref{app:proofs}.
\end{proof}

A direct corollary of Proposition~\ref{prop: lim_n_to_inf} and Conjecture~\ref{thm3} with capped delegation is that
 $\lim_{n\rightarrow \infty}P_c(n,p,m)=1$ for any fixed $p>0.5$, any $c\in \mathbb{N}$ and any $m$, where $m$ can even depend on $n$.

Next, we show that the probability that $A$ wins with free delegation converges to the probability that $A$ wins with conventional voting, if $m$ goes to infinity.

\begin{proposition}\label{P_lim}
$\lim_{m\rightarrow \infty}P(p,m)=P(p)$ for any $p\in(0,1)$.
\end{proposition}
\begin{proof}
See Appendix \ref{app:proofs}.
\end{proof}

Another direct corollary of Proposition~\ref{P_lim} and Conjecture~\ref{thm3} for capped delegation is that $\lim_{m\rightarrow \infty}P_c(p,m)=P(p)$ for any $c\in \mathbb{N}$ and $p\in(0,1)$.


{\section{Weighted Voting}

Let $P_n(p)$ denote the probability that the majority implements the right alternative with $n \in \mathbb{N}$ voters, where each receives the right signal with probability $p$. Let $P_{n,w}(p)$ denotes the probability that the majority implements the right alternative with $n$ voters, where each receives the right signal with probability $p$ and $w:=(w_1,\cdots, w_n)$ denotes the vector of weights of all voters. Each weight $w_i$ is a non-negative real number. We prove the following generic result in the information aggregation setting: 

\begin{theorem}\label{generic}
$P_n(p)\geq P_{n,w}(p)$ for any $p\in[0,1],n\in \mathbb{N}$ and $w = (w_1,...,w_n)$ where $w_i\geq 0$ for each $i=1,...,n$. 
\end{theorem}

This result suggests that the probability of implementing the right achieves maximum when all weights are equal. In fact, from the proof we note that weight vectors where no minority subset outweighs the complementary majority subset maximizes this probability.  

}

\section{Discussion and Conclusion}
We showed that the introduction of vote delegation leads to lower probabilities of winning for the ex-ante majority. Numerical simulations show that free delegation leads to even lower probabilities than capped delegation. However, both delegation processes lead to lower probabilities than conventional voting. These results are particularly important in blockchain governance if one wants to prevent dishonest agents from increasing their probability of winning. Although the setting we analyze in the paper is as simple as possible, we already obtain non-trivial observations. 

In addition to the formal results, we also conjecture that $P(n,p,m)$ is first decreasing and then increasing in $m$, based on numerical evidence.
This conjecture is interesting in the context of delegation as a strategic decision. For the delegators, whether or not they belong to the ex-ante majority, choosing between delegating and abstaining becomes a strategic decision. If we prove only one part of the conjecture, namely that the winning probability is increasing from some point on, we will obtain that in the equilibrium, all majority voters are delegating for a large enough number of delegators $m$. First, note that no one delegating cannot be an equilibrium state, because minority voters prefer to delegate. Once some of them delegate, the probability that the majority of voters is winning is smaller than the one in conventional voting. Using Proposition~\ref{pr_conv}, we obtain that with a large number of delegators, all majority voters delegating is sustainable in the equilibrium. 
We leave this challenging issue to future research.

\subsection*{Acknowledgements}
This research was partially supported by the Zurich Information Security and Privacy Center (ZISC) at ETH Zurich.

\bibliographystyle{apalike}
\bibliography{references}

\clearpage

\appendix
\section{Proofs} \label{app:proofs}

Before proving Theorem~\ref{thm}, we show the following lemma:
\begin{lemma}\label{lemma_G_sum_1}
For any $l,t,m \in \mathbb{N}$ we have $G(l+t,l,m) + G(l,l+t,m)=1$.
\end{lemma}
\begin{proof}[Proof of Lemma~\ref{lemma_G_sum_1}]

Similarly to the definition of $G$, let $H(k,l,m)$ denote the probability that $B$ wins if $A$ voters have $k$ votes, $B$ voters have $l$ votes and $m$ voters are delegating. Then, 

\begin{equation}
    G(l+t, l,m) = H(l, l+t, m) = 1 - G(l, l+t, m),
\end{equation}

where the first equality follows from the symmetry, and the second equality follows from the fact that events of $A$ and $B$ winning are complementary. 
\end{proof}

With this Lemma we proceed to prove Theorem~\ref{thm}.

\begin{proof}[Proof of Theorem~\ref{thm}]
To prove that $P(p,m)$ is smaller than $P(p)$ for $m>0$ and $p>0.5$, we take a closer look at the summands in $P(p,m)$ and $P(p)$, in~\eqref{pr_cap_del} and~\eqref{pr_conv}, respectively. We prove the inequality for each $i$. We use two properties. 

First, from the Lemma~\ref{lemma_G_sum_1} we know that $G(k, i - k, m) + G(i - k, k ,m) = 1$. Second, 

\begin{equation}\label{binomial_ineq}
    P[Bin(i,p) = k] \geq P[Bin(i,p) = i - k]
\end{equation} for any $k\geq i-k$ and $p\geq \frac{1}{2}$. Since ${i \choose k}={i\choose i-k}$, the condition is equivalent to $p^k(1-p)^{i-k}\geq p^{i-k}(1-p)^k$, which is equivalent to $\frac{p}{1-p}^{2k-i}\geq 1$. For $k\geq i-k$, we have $2k-i\geq 0$. For $p\geq \frac{1}{2}$ the condition holds. 

The difference $P(p)-P(p,m)$ is equal to 

\begin{equation}
    \sum_{i=0}^{\infty}F(i)\underbrace{\sum_{k=0}^{i}{i \choose k}p^{k}(1-p)^{i-k}(g(k,i-k)-G(k,i-k,m))}_{\Delta_i}.
\end{equation}
When $k<i-k$, we have that $g(k,i-k)-G(k,i-k,m)=-G(k,i-k,m)$. When $k>i-k$, we have $g(k,i-k)-G(k,i-k,m) = G(i-k, k,m)$. When $k = i-k$ we have that $g(k,i-k)-G(k,i-k,m) = 0$. We split $\Delta_i$ into two sums. One, where $k>i-k$ and one where $k<i-k$.
\begin{align*}
    \Delta_i &= \sum_{k=0}^{i}{i \choose k}p^{k}(1-p)^{i-k}(g(k,i-k)-G(k,i-k,m)) \\
    &= \sum_{k = \lceil \frac{i+1}{2} \rceil}^i {i \choose k}p^{k}(1-p)^{i-k}G(i-k,k,m) - \sum_{k = 0}^{\lfloor \frac{i-1}{2} \rfloor} {i \choose k} p^{k}(1-p)^{i-k}G(k,i-k,m)\\
    & \overset{(*)}{=} \sum_{k = \lceil \frac{i+1}{2} \rceil}^i {i \choose k}p^{k}(1-p)^{i-k}G(i-k,k,m) - \sum_{k = \lceil \frac{i+1}{2} \rceil}^i {i \choose i-k} p^{i-k}(1-p)^{k}G(i-k,k,m)\\
    &= \sum_{k = \lceil \frac{i+1}{2} \rceil}^i \underbrace{\bigg[{i \choose k}p^{k}(1-p)^{i-k} - {i \choose i-k} p^{i-k}(1-p)^{k} \bigg]}_{>0 \text{ by \eqref{binomial_ineq}}} G(i-k,k,m)\\
    &> 0.
\end{align*}
In step $(*)$ we made a change of variables in the second sum from $k \to i-k$.
Hence, $P(p)-P(p,m)>0$. This finishes the proof of the theorem.
\end{proof}

Next, we prove the following two lemmata:

\begin{lemma}\label{lemma_G_cap_1} 
For $c\geq 2$, $l,m \in \mathbb{N}$ and $t \in \{1,...,m\}$:
    \begin{equation*}
        G_c(l+t,l,m) + G_c(l,l+t,m) = 1.
    \end{equation*}
\end{lemma}

\begin{proof}[Proof of Lemma \ref{lemma_G_cap_1}]

Similarly to the definition of $G_c$, let $H_c(k,l,m)$ denote the probability that $B$ wins if $A$ voters have $k$ votes, $B$ voters have $l$ votes, $m$ voters are delegating and the cap is equal to $c$. Then, 

\begin{equation}
    G_c(l+t, l,m) = H_c(l, l+t, m) = 1 - G_c(l, l+t, m),
\end{equation}

where the first equality follows from the symmetry, and the second equality follows from the fact that events of $A$ and $B$ winning are complementary. 
\end{proof}

\begin{lemma}\label{symmetric_GC}
$G_c(l,l,m)=\frac{1}{2}$ for any $c,l,m\in \mathbb{N}$ .
\end{lemma} 

\begin{proof}[Proof of Lemma \ref{symmetric_GC}]

Proof follows immediately from the symmetry. 
\end{proof}

Next, we prove Theorem~\ref{thm2}.

\begin{proof}[Proof of Theorem~\ref{thm2}]
As for the proof of Theorem~\ref{thm}, the key lies in Lemma~\ref{lemma_G_cap_1}. We proceed as in the proof of Theorem~\ref{thm}.
We use two properties. 

First, from Lemma~\ref{lemma_G_cap_1} we know that $G_c(k, i - k, m) + G_c(i - k, k ,m) = 1$. Second, 

\begin{equation}\label{binomial_ineq2}
    P[Bin(i,p) = k] \geq P[Bin(i,p) = i - k]
\end{equation} for any $k\geq i-k$ and $p\geq \frac{1}{2}$. Since ${i \choose k}={i\choose i-k}$, the condition is equivalent to $p^k(1-p)^{i-k}\geq p^{i-k}(1-p)^k$, which is equivalent to $\frac{p}{1-p}^{2k-i}\geq 1$. For $k\geq i-k$, we have $2k-i\geq 0$. For $p\geq \frac{1}{2}$ the condition holds. 

The difference $P(p)-P_c(p,m)$ is equal to 

\begin{equation}
    \sum_{i=0}^{\infty}F(i)\sum_{k=0}^{i}{i \choose k}p^{k}(1-p)^{i-k}(g(k,i-k)-G_c(k,i-k,m)).
\end{equation}

When $k<i-k$, we have that $g(k,i-k)-G_c(k,i-k,m)=-G_c(k,i-k,m)$. Similarly, we have $g(i-k,k)-G_c(i-k,k,m) = G_c(k, i-k,m)$. On the other hand, from~\eqref{binomial_ineq2} we know that the coefficient with $-G_c(k,i-k,m)$ is smaller than the coefficient with $G_c(k,i-k,m)$. Therefore, $P(p)-P_c(p,m) > 0$. This finishes the proof of the theorem. 
\end{proof}

\begin{proof}[Proof of Proposition \ref{prop: lim_p_to_1}]

We can verify that $P(1,m)=1$. 
\begin{equation*}
    P(1,m) = \sum_{i=0}^\infty F(i) \underbrace{g(i+m,0)}_{=1} = 1.
\end{equation*}
Indeed, only the terms where $k=i$ and $h=m$ survive.
Since the function $P(p,m)$ is uniformly continuous in $p$, the proposition is proved.
\end{proof}

\begin{proof}[Proof of Proposition \ref{prop: lim_n_to_inf}]

Let $p>0.5$ and $\epsilon \in (0,p-\frac{1}{2})$. Define two Poisson random variables $K$, with parameter $np$, and $L$, with parameter $n(1-p)$.
The probability that $A$-voters have at least $n(\frac{1}{2}+\epsilon)$ votes can be bounded by using the Poisson random variable concentration inequalities from ~\cite{Mitzenmacher} which say that for a Poisson random variable $X$ with parameter $\lambda$: 

If $x < \lambda$, then
\begin{equation}
\label{low_concentration}
P(X\leq x)\leq \frac{e^{-\lambda}(e\lambda)^x}{x^x},
\end{equation}
and if  $x > \lambda$, then
\begin{equation}
\label{high_concentration}
P(X\geq x)\leq \frac{e^{-\lambda}(e\lambda)^x}{x^x}.
\end{equation}

Using~\eqref{low_concentration} for $n(\frac{1}{2}+\epsilon)$ and $np$, we obtain:

\begin{equation*}
    P[K<n(\frac{1}{2}+\epsilon)] \leq \frac{(enp)^{n(\frac{1}{2}+\epsilon)}}{e^{np} (n(\frac{1}{2}+\epsilon))^{n(\frac{1}{2}+\epsilon)}} = \left( \frac{\left(\frac{p}{1+\epsilon}\right)^{\frac{1}{2}+\epsilon}}{e^{p-\frac{1}{2}-\epsilon}}\right)^n \overset{n \to \infty}{\longrightarrow} 0.
\end{equation*}

The last implication holds if we show that $$\left(\frac{p}{q}\right)^q<e^{p-q},$$

where $q:=\frac{1}{2}+\epsilon$. Consider $f(p)=e^{p-q}-(p/q)^q$. Then, $f(q)=0$ and $f$ is increasing in $p$. The latter holds because $\frac{\partial f(p)}{\partial p}=e^{p-q}-(p/q)^{q-1}>0$ for any $p>q$, since $e^{p-q}>1$ and $(p/q)^{q-1}<1$. 

Hence, the probability that $A$-voters have at least $n(\frac{1}{2}+\epsilon)$ votes is
\begin{equation*}
    P[K\geq n(\frac{1}{2}+\epsilon)] = 1-P[K<n(\frac{1}{2}+\epsilon)] \overset{n \to \infty}{\longrightarrow} 1.
\end{equation*}

At the same time we apply~\eqref{high_concentration} for $n(\frac{1}{2}-\epsilon)$ and $n(1-p)$, we obtain:
\begin{equation}\label{eq: l>n(0.5-eps)}
    P[L>n(\frac{1}{2}-\epsilon)] \leq \frac{(en(1-p))^{n(\frac{1}{2}-\epsilon)}}{e^{n(1-p)} (n(\frac{1}{2}-\epsilon))^{n(\frac{1}{2}-\epsilon)}} \overset{n \to \infty}{\longrightarrow} 0.
\end{equation}

The last implication follows from the following: We can rewrite the fraction on the right-hand side as follows
\begin{equation*}
     \frac{e^{n(\frac{1}{2}-\epsilon)(1+\log(n(1-p)))}}{e^{n(\frac{1}{2}-\epsilon) (\frac{(1-p)}{(\frac{1}{2}-\epsilon)} + \log(n(\frac{1}{2}-\epsilon)))}}   = \exp\big( n(\frac{1}{2}-\epsilon) \underbrace{\big[ 1+\log(n(1-p)) - \frac{1-p}{\frac{1}{2}-\epsilon} - \log(n(\frac{1}{2}-\epsilon)) \big]}_{=:g} \big).
\end{equation*}
Note that $g$ is independent of $n$ as we can write
\begin{equation*}
    g = 1 - \frac{1-p}{\frac{1}{2}-\epsilon} + \log(\frac{1-p}{\frac{1}{2}-\epsilon}).
\end{equation*}
By assumption on $p$ and $\epsilon$ we have that 
\begin{equation*}
     \frac{1-p}{\frac{1}{2}-\epsilon} < 1.
\end{equation*}
We consider the following function for any $y \in (0,1)$,
\begin{equation*}
    f(y):= 1 -y+ \log(y).
\end{equation*}
Then $f(\frac{1-p}{\frac{1}{2}-\epsilon}) = g$.
Function $f$ has the following properties: First, $\lim_{y \to 0} f(y) = -\infty$ and $\lim_{y\to 1} f(y) = 0$. Second, the derivative $f'(y) = -1 + \frac{1}{y} > 0$ since $y<1$. Hence, $f(y)$ is negative for any $y<1$. This implies that $g<0$ and hence the right-hand side of Equation \eqref{eq: l>n(0.5-eps)} is $\exp(n(\frac{1}{2}-\epsilon)g)$ and converges to $0$ for $n \to \infty$ since $g$ is negative.




Hence, the probability that $B$-voters have at most $n(\frac{1}{2}-\epsilon)$ votes is
\begin{equation*}
    P[L\leq n(\frac{1}{2}-\epsilon)] = 1-P[L > n(\frac{1}{2}-\epsilon)] \overset{n \to \infty}{\longrightarrow} 1.
\end{equation*}

These two results and the fact that $K, L$ are independent yield: 
\begin{equation*}
    P[K\geq n(\frac{1}{2}+\epsilon) \text{ and } L\leq n(\frac{1}{2}-\epsilon)] = \underbrace{P[K\geq n(\frac{1}{2}+\epsilon)]}_{\overset{n \to \infty}{\longrightarrow} 1} \underbrace{P[L\leq n(\frac{1}{2}-\epsilon)]}_{\overset{n \to \infty}{\longrightarrow} 1} \overset{n \to \infty}{\longrightarrow} 1.
\end{equation*}
It follows that for $m< 2\epsilon n$, $A$-voters are in the majority with probability $1$, if $n$ goes to infinity.

Next, we consider $m>2\epsilon n$.
From above, we know that for high enough $n$, with probability $1$ $K\geq n(\frac{1}{2}+\epsilon)$ and $L \leq n(\frac{1}{2}-\epsilon)$. Hence, 
\begin{align*}
    \frac{K}{L} \geq \frac{\frac{1}{2}+\epsilon}{\frac{1}{2}-\epsilon} > 1.
\end{align*}

We define the i.i.d. random variables:

\begin{equation*}
    X_i : = \begin{cases}
+1 & \text{w.p. } \frac{1}{2}+\epsilon \\
-1 & \text{w.p. } \frac{1}{2}-\epsilon
\end{cases}
\end{equation*}
for  $i=1,\cdots,m$. The surplus of delegated votes to party $A$-voters is the sum of the $X_i$. We know that $E[X_i] = 2 \epsilon >0$ and by the law of large numbers, $\frac{1}{m}\sum_{i=1}^m X_i \overset{m \to \infty}{\longrightarrow} 2\epsilon >0$. As $m>2\epsilon n$, $n \to \infty$ implies $m \to \infty$.
\end{proof}



\begin{proof} [Proof of Proposition \ref{P_lim}]
First, using Hoeffding's inequality\footnote{See \cite{Hoeffding}.}, we can show that $\lim_{m \to \infty} G(k,l,m)=g(k,l)$ for any fixed $k$ and $l$.
\begin{itemize}
    \item Consider fixed pair $k,l$, so that $k=l\neq 0$. Then, with $q := \frac{k}{k+l}=\frac{1}{2}$, we have
\begin{align*}
    G(k,l,m) &= \sum_{h = \lfloor \frac{m}{2}+1 \rfloor}^m \binom{m}{h} q^h (1-q)^{m-h} + \frac{1}{2} \binom{m}{\frac{m}{2}} q^{\frac{m}{2}} (1-q)^{\frac{m}{2}}\\
    &= \left(\frac{1}{2}\right)^m \sum_{h = \lfloor \frac{m}{2}+1 \rfloor}^m \binom{m}{h} + \left(\frac{1}{2}\right)^{m+1} \binom{m}{\frac{m}{2}}.
\end{align*}
Note that the latter summand is zero if $m$ is odd.
\begin{itemize}
    \item If $m$ is odd: 
    \begin{align*}
         G(k,l,m) = \left(\frac{1}{2}\right)^m \sum_{h = \lfloor \frac{m}{2}+1 \rfloor}^m \binom{m}{h}= \left(\frac{1}{2}\right)^m \frac{2^m}{2} = \frac{1}{2}.
    \end{align*}
    \item If $m$ is even: 
        \begin{align*}
         G(k,l,m) = \left(\frac{1}{2}\right)^m \left[ \sum_{h = \lfloor \frac{m}{2}+1 \rfloor}^m \binom{m}{h} + \left(\frac{1}{2}\right) \binom{m}{\frac{m}{2}} \right]= \left(\frac{1}{2}\right)^m \frac{2^m}{2} = \frac{1}{2}.
    \end{align*}
\end{itemize}
Hence, for $k=l$, we have $G(k,l,m) = g(k,l)=\frac{1}{2}$. If $k=l=0$, then by definition $G(0,0,m)=g(0,0)=\frac{1}{2}$.

\item Consider fixed pair $k,l$, so that $k<l$. Then, with $q := \frac{k}{k+l}<\frac{1}{2}$, we have
\begin{align*}
    G(k,l,m) &= \sum_{h = \lfloor \frac{l-k+m}{2}+1 \rfloor}^m \binom{m}{h} q^h (1-q)^{m-h} + \frac{1}{2} \binom{m}{\frac{l-k+m}{2}} q^{\frac{l-k+m}{2}} (1-q)^{m-\frac{l-k+m}{2}}.
\end{align*}
Define $a:= \frac{l-k+m}{2}$. The first summand above is equal to the probability $P[X > a]$ for a random variable $X \sim Bin(m,q)$. We consider $P[X\geq a]$, which is $G(k,l,m)$ plus a non-negative term (the second term of $G(k,l,m)$) and show that this converges to $0$ and hence also $G(k,l,m)$ converges to $0$.

Hoeffding's inequality  yields:
\begin{equation*}
    P[X \geq m (q+\epsilon)] \leq \exp(-2\epsilon^2 m).
\end{equation*}
We calculate $\epsilon$ by solving $m(q+\epsilon) = a$:
\begin{align*}
 \epsilon = \frac{l-k+m}{2m}-q.
\end{align*}

Let $t:=l-k$, then by Hoeffding's inequality:
\begin{align*}
    P[X \geq a] &\leq \exp(-2 (\frac{l-k+m}{2m}-q)^2 m)\\
    &= \exp(-m(2q^2-2q+\frac{1}{2}) -t-\frac{t^2}{2m}+2qt).
\end{align*}
The right-hand side goes to $0$ for $m \to \infty$ if $2q^2-2q+\frac{1}{2}>0$. This inequality is true for any $q \neq \frac{1}{2}$. As $q<\frac{1}{2}$, the RHS goes to $0$ for  $m \to \infty$. That is,
\begin{equation*}
    \lim_{ m \to \infty } P[X \geq a] = 0.
\end{equation*}
Hence, for $k<l$ we have $\lim_{m \to \infty}G(k,l,m) = g(k,l)= 0$.

\item Consider fixed pair $k,l$, so that $k>l$. Then, with $q := \frac{k}{k+l}>\frac{1}{2}$ we have
\begin{align*}
    G(k,l,m) &= \sum_{h = \lfloor \frac{l-k+m}{2}+1 \rfloor}^m \binom{m}{h} q^h (1-q)^{m-h} + \frac{1}{2} \binom{m}{\frac{l-k+m}{2}} q^{\frac{l-k+m}{2}} (1-q)^{m-\frac{l-k+m}{2}}.
\end{align*}
Define $a:= \frac{l-k+m}{2}$. The first summand above is equal to the probability $P[X > a]$ for a random variable $X \sim Bin(m,q)$. As before, we consider the value $P[X \geq a]$.

Hoeffding's inequality yields:
\begin{equation*}
    P[X < m (q-\epsilon)] \leq \exp(-2\epsilon^2 m).
\end{equation*}
We calculate $\epsilon$ by solving $m(q-\epsilon) = a$:
\begin{align*}
 \epsilon = q - \frac{l-k+m}{2m}.
\end{align*}

Let $t:=|l-k|$, then, by Hoeffding's inequality:
\begin{align*}
    P[X < a] &\leq \exp(-2 (q-\frac{l-k+m}{2m})^2 m)\\
    &= \exp(-m(2q^2-2q+\frac{1}{2}) -t-\frac{t^2}{2m}+2qt).
\end{align*}
The RHS of the latter converges to $0$ for $m \to \infty$ if $2q^2-2q+\frac{1}{2}>0$. This inequality is true for any $q \neq \frac{1}{2}$. As $q>\frac{1}{2}$, the RHS converges to $0$ for  $m \to \infty$. That is,
\begin{equation*}
    \lim_{ m \to \infty } P[X < a] = 0.
\end{equation*}
Therefore, 
\begin{equation*}
    \lim_{ m \to \infty } P[X \geq a] = \lim_{ m \to \infty } (1-P[X < a]) = 1.
\end{equation*}
Hence, for $k>l$, we have $\lim_{m \to \infty} G(k,l,m) = g(k,l)= 1$.
\end{itemize}
It follows that $\lim_{m \to \infty} G(k,l,m)=g(k,l)$ for any fixed $k$ and $l$.\\\\
In the next step, we want to show that $\lim_{m \to \infty} P(p,m)=P(p)$.

{Let $Q(k,p)$ denote the probability that $k$ voters are $X$-voters, with $X\in \{A,B\}$ and each voter being $X$-voter with probability $p$.
Then, 
\begin{equation*}
    Q(k,p) = \sum_{i = k}^\infty F(i) \binom{i}{k} p^{k} (1-p)^{i-k}.
\end{equation*}
If $F$ is a Poisson distribution with parameter n, then 
\begin{align*}
    Q(k,p) &= \sum_{i=k}^\infty \frac{n^i}{e^n i!} \binom{i}{k}p^k(1-p)^{i-k} = \frac{(np)^k}{e^{np}k!}\sum_{i=0}^\infty \frac{(n(1-p))^{i-k} k!}{e^{n(1-p)}i!} \binom{i}{k} \\ &= \frac{(np)^k}{e^{np}k!} \sum_{i=0}^\infty \frac{(n(1-p))^{i-k}}{e^{n(1-p)}(i-k)!} = \frac{(np)^k}{e^{np}k!}.
\end{align*}
}

Define, 
\begin{equation*}
    a_{k}(m) := \sum_{l=0}^{\infty} Q(k,p) Q(l,1-p) G(k,l,m).
\end{equation*}
Then, 
\begin{equation*}
    P(p,m) = \sum_{k=0}^{\infty} a_{k}(m).
\end{equation*}
Note that, as $G(k,l,m)\leq 1$,
\begin{equation*}
    |a_k(m)| \leq Q(k,p) \sum_{l=0}^{\infty} Q(l,1-p) = Q(k,p) =: M_k,
\end{equation*}
and $\sum_{k=0}^{\infty} M_k =1 < \infty$. Then, by dominated convergence theorem,
\begin{align*}
    \lim_{m \to \infty} P(p,m) &= \sum_{k=0}^{\infty} \lim_{m \to \infty} a_{k}(m)\\
    &= \sum_{k=0}^{\infty} Q(k,p) \lim_{m \to \infty} \sum_{l=0}^{\infty} Q(l,1-p) G(k,l,m).
\end{align*}
Define 
\begin{equation*}
    b_l(m) := Q(l,1-p) G(k,l,m).
\end{equation*}
Note that
\begin{equation*}
    |b_l(m)| \leq Q(l,1-p) = : M_l
\end{equation*}
and $\sum_{l=0}^{\infty} M_l= 1 < \infty$.
Hence, by dominated convergence theorem,
\begin{align*}
    \lim_{m \to \infty} \sum_{l=0}^{\infty} b_l(m) &= \sum_{l=0}^{\infty}\lim_{m \to \infty} b_l(m) \\
    &= \sum_{l=0}^{\infty} Q(l,1-p) \lim_{m \to \infty} G(k,l,m).
\end{align*}
Taken together, by applying dominated convergence theorem twice, we obtain:
\begin{align*}
    \lim_{m \to \infty} P(p,m) &= \sum_{k=0}^{\infty} \lim_{m \to \infty} a_{k}(m)\\
    &= \sum_{k=0}^{\infty} Q(k,p) \lim_{m \to \infty} \sum_{l=0}^{\infty} b_l(m)\\
    &= \sum_{k=0}^{\infty} Q(k,p) \sum_{l=0}^{\infty} \lim_{m \to \infty} b_l(m)\\
    &= \sum_{k=0}^{\infty} Q(k,p) \sum_{l=0}^{\infty} Q(l,1-p) \lim_{m \to \infty} G(k,l,m)\\
    &= \sum_{k=0}^{\infty} Q(k,p) \sum_{l=0}^{\infty} Q(l,1-p) g(k,l) = P(p).
\end{align*}
\end{proof}


{
\begin{proof}[Proof of Theorem~\ref{generic}]
    Let $N = \{1,...,n\}$ be the set of $n\in \mathbb{N}$ voters.
    First we assume that $n$ is an odd number. Let $w(S)$ denote the sum of weights in the set $S\subseteq N$, that is $w(S):=\sum_{i\in S}w_i$ and assume that no subset $S \subseteq N$ has the weight equal to the weight of the complementing subset $\bar{S} :=N\setminus S$, that is, $\forall S \subseteq N, w(S) \neq w(\bar{S})$. Note that $P_n(p)$ can be calculated in the following way:
    
    \begin{equation*}
         P_n(p) = \sum_{S\subseteq N, |S|>n/2}p^{|S|}(1-p)^{n-|S|}. 
    \end{equation*}
    
    Furthermore, $P_{n,w}(p)$ is calculated in the following way: 
    
    \begin{equation*}
        P_{n,w}(p) = \sum_{S\subseteq N, w(S)>w(N)/2}p^{|S|}(1-p)^{n-|S|}.
    \end{equation*}
    
    For any subset $S$, so that $w(S)>w(N)/2$, we have the following two possibilities: 
    
    \begin{enumerate}
        \item In the first case, it holds that $|S|>n/2$. In this case, the corresponding summand in the formula of $P_n(p)$ is the same as in $P_{n,w}(p)$.
        \item In the second case, it holds that $|S|<n/2$. In this case, $|\bar{S}| = n-|S| >n/2$. The summand for $\bar{S}$ in the formula of $P_n(p)$, namely $p^{|\bar{S}|}(1-p)^{|S|}$ is strictly larger than the summand in the formula of $P_{n,w}(p)$ that corresponds to $S$, namely $p^{|S|}(1-p)^{|\bar{S}|}$. In the latter claim we use two properties: $p>\frac{1}{2}$ and $|S|<|\bar{S}|$. 
    \end{enumerate}
    
    That is, for each summand in the formula of $P_{n,w}(p)$, there is a summand in the formula of $P_n(p)$ which is the equal or larger. The case with even $n$ and/or the weight vectors $w$ which have subsets with equal weight to their complement subsets is similar to the case considered in the theorem proof.   
    
\end{proof}
}

\end{document}